\documentclass[12pt]{article}
\usepackage{amssymb}
\usepackage{amsfonts}
\usepackage{amsmath}

\newtheorem{theorem}{Theorem}

\newtheorem{definition}{Definition}

\newtheorem{lemma}{Lemma}

\newtheorem{remark}{Remark}

\begin{document}

\title{Stable Improvement Cycle Mechanism Versus Efficiency-adjusted
Deferred Acceptance Mechanism beyond Weak Priority Orders\thanks{%
This paper builds on some of the results contained in an earlier manuscript
entitled \textquotedblleft School Choice with Multiple
Priorities,\textquotedblright\ while adding new results and substantially
revising. This work was supported by JSPS KAKENHI Grant Number 25K05004.}}
\author{Minoru Kitahara\thanks{%
Department of Economics, Osaka Metropolitan University} \and Yasunori Okumura%
\thanks{%
Corresponding author. Department of Logistics and Information Engineering,
TUMSAT, 2-1-6, Etchujima, Koto-ku, Tokyo, 135-8533 Japan.
Phone:+81-3-5245-7300. Fax:+81-3-5245-7300. E-mail: yokumu0@kaiyodai.ac.jp}}
\maketitle

\begin{center}
\textbf{Abstract}
\end{center}

This paper compares the stable improvement cycles (SIC) mechanism of Erdil
and Ergin (2008) with the efficiency-adjusted deferred acceptance (EADA)
mechanism of Kesten (2010) in school choice problems with incomplete
priorities. When school priorities are strict partial orders, we show that
the EADA mechanism produces a constrained efficient matching. Combined with
earlier results for the SIC mechanism, this establishes the outcome
equivalence of the two mechanisms: the sets of matchings attainable under
them coincide. We then move beyond transitivity and consider acyclic
priority relations. On this broader domain, neither mechanism is guaranteed
to produce a constrained efficient matching, but their failures take
opposite forms. The SIC mechanism may under-improve: its outcome is always
stable but may be Pareto dominated by another stable matching. By contrast,
the EADA mechanism may over-improve: its outcome is never Pareto dominated
by any stable matching but may itself be unstable.


\newpage

\section{Introduction}

A central theme in the school-choice literature is the trade-off between
stability and student welfare; see, for example, Abdulkadiro\u{g}lu and S%
\"{o}nmez (2003), Erdil and Ergin (2008), Abdulkadiro\u{g}lu et al. (2009),
and Kesten (2010). More specifically, stability requires schools' priorities
to be respected, but this requirement limits the extent to which assignments
can be determined in accordance with students' preferences. This perspective
suggests that student welfare may be improved not only by relaxing existing
priority constraints but also by avoiding artificial refinements of
priorities that introduce unnecessary constraints.

These observations led to two influential mechanisms: Kesten's (2010)
efficiency-adjusted deferred acceptance mechanism, or EADA, and Erdil and
Ergin's (2008) stable improvement cycles mechanism, or SIC. Although these
mechanisms differ substantially in their procedures, they are known to
produce closely related and, in some cases, identical outcomes. By reviewing
existing results and establishing new ones, this paper compares the two
mechanisms and clarifies when and in what sense they coincide and when they
diverge.

To clarify the relationship between the two mechanisms, it is essential to
consider a broader class of school priorities than is commonly assumed. In
the classical matching model of Gale and Shapley (1962), the rankings on
both sides are interpreted as preferences. It is therefore natural to assume
that schools' rankings of students are linear orders.

In school choice, by contrast, the school-side ordering is not a preference
of the school. Rather, it is a priority relation generated by policy rules
specified by the school-choice system. Such rules may assign equal priority
to multiple students. For example, students who have the same walk-zone and
sibling status may be tied in a school's priority relation. This possibility
is central to Erdil and Ergin (2008) and Abdulkadiro\u{g}lu et al. (2009),
who consider school choice problems in which school priorities are
represented by weak orders and may therefore contain ties.

A stable matching obtained after breaking such ties is also stable under the
original weak priorities. However, it need not be constrained efficient: it
may be Pareto dominated by another matching that is stable under the
original priorities. Thus, artificially breaking ties may introduce
unnecessary priority constraints and lead to an inefficient outcome within
the original stable set. This observation motivated mechanisms that select
constrained efficient matchings, namely, stable matchings that are not
Pareto dominated by any other stable matching.

Erdil and Ergin (2008) introduce the SIC mechanism to obtain a constrained
efficient matching when school priorities are weak orders. The mechanism has
subsequently been applied or extended in several related contexts, including
Kesten and \"{U}nver (2015), Dur et al. (2019), Erdil and Kumano (2019),
Kitahara and Okumura (2020, 2021), and Bando et al. (2025).

Kesten (2010) introduces the EADA mechanism to improve efficiency by
allowing students to consent to certain priority violations. He also shows
that, when school priorities are weak orders, the EADA mechanism can be used
to obtain a constrained efficient matching. The mechanism has since been
applied in a broader range of settings; see Cerrone et al. (2024) for an
overview of this literature.

More recently, several studies have moved beyond weak priorities and
considered settings in which school priorities allow richer forms of
incomparability among students. Examples include Che et al. (2019a,b), Dur
et al. (2019), Kitahara and Okumura (2020, 2021, 2024), and Kuvalekar
(2023). This paper also discusses examples involving potentially important
school-choice settings in which priority relations cannot be represented as
weak orders, thereby providing further motivation for studying richer
priority domains.

Among these studies, Kitahara and Okumura (2021) extend the applicability of
the SIC mechanism to priorities represented by strict partial orders. They
show that, starting from any stable matching that is not constrained
efficient, the mechanism yields a constrained efficient matching that Pareto
dominates the initial matching. Building on this line of research, the
present paper establishes an analogous result for EADA: when priorities are
strict partial orders, EADA also yields a constrained efficient matching.
Combined with the characterization of Kitahara and Okumura (2024), this
result implies that the sets of outcomes attainable under the SIC and EADA
mechanisms coincide.

The class of strict partial order priorities can be further generalized by
relaxing transitivity while retaining acyclicity. Acyclic but non-transitive
priorities are not merely a theoretical possibility. They arise, for
example, in the general consent model introduced later in this paper, which
extends the models studied by Kesten (2010) and Dur et al. (2019). In this
model, students may waive their priority claims in favor of particular other
students. As we discuss in detail later, without appropriate restrictions on
consent, the induced priority relations may remain acyclic but fail to be
transitive.

Although stable matchings continue to exist under acyclic priorities, the
constrained efficiency results for SIC and EADA do not extend to this
broader domain. We provide two examples with acyclic but non-transitive
priorities. The first shows that the SIC mechanism may fail to yield a
constrained efficient matching, while the second establishes an analogous
failure for EADA. The two mechanisms, however, fail in opposite ways. The
SIC mechanism may exhibit \textquotedblleft
under-improvement\textquotedblright : its outcome remains stable but is
Pareto dominated by another stable matching. By contrast, the EADA mechanism
may exhibit \textquotedblleft over-improvement\textquotedblright : no stable
matching Pareto dominates its outcome, but the outcome itself is not stable.
These contrasting failures clarify the roles played by transitivity in the
two mechanisms.

\section{Model}

Let $I$ and $S$ be the sets of students and schools, respectively, where $%
\left\vert I\right\vert \geq 3$. Each student $i\in I$ has a linear order on 
$S\cup \left\{ \emptyset \right\} $ denoted by $P_{i}$ representing the
preference order of student $i$, where $sP_{i}s^{\prime }$ (or equivalently $%
\left( s,s^{\prime }\right) \in P_{i}$) means that $i$ prefers $s\in S\cup
\left\{ \emptyset \right\} $ to $s^{\prime }\in S\cup \left\{ \emptyset
\right\} $ and $\emptyset $ represents their best outside option. If $%
sP_{i}\emptyset $, then school $s$ is said to be \textbf{acceptable} for $i$%
. Further, $sR_{i}s^{\prime }$ means $sP_{i}s^{\prime }$ or $s=s^{\prime }$.
Let $\mathbb{P}$ be the set of all possible linear orders on $S\cup \left\{
\emptyset \right\} $.

Each school $s$ has a capacity constraint represented by $q_{s}\in \mathbb{Z}%
_{++}$ and moreover, $q=\left( q_{s}\right) _{s\in S}.$ Let $\succ _{s}$ be
an asymmetric binary relation on $I$ representing the priority for school $s$%
, where $i\succ _{s}j$ (or equivalently $\left( i,j\right) \in \succ _{s}$)
means that $i$ has a higher priority than $j$ for school $s$. As discussed
in detail in Section 2.1, unlike students' preference orders, schools'
priority relations are not necessarily linear orders. Let $\mathcal{B}$ be
the set of all possible asymmetric binary relations on $I$. Hereafter, we
call $\succ _{s}\in \mathcal{B}$ a \textbf{priority relation}. Let $\succ
=\left( \succ _{s}\right) _{s\in S}$.

We let 
\begin{equation*}
G=\left( I,S,P,\succ ,q\right)
\end{equation*}%
be a \textbf{school choice problem}.

A \textbf{matching} $\mu $ is a mapping satisfying $\mu (i)\in S\cup
\{\emptyset \},$ $\mu \left( s\right) \subseteq I$, and $\mu (i)=s$ if and
only if $i\in \mu \left( s\right) ,$ and $\left\vert \mu \left( s\right)
\right\vert \leq q_{s}$ for all $s\in S$. Note that $\mu (i)=\emptyset $
means that $i$ is unmatched to any school and $\mu (i)=s\in S$ means that $i$
is matched to $s$ under a matching $\mu $.

A matching $\mu $ is said to be \textbf{individually rational} if $\mu
\left( i\right) R_{i}\emptyset $ for all $i\in I$. A matching $\mu $ is said
to be \textbf{non-wasteful} if $sP_{i}\mu \left( i\right) $ implies $%
\left\vert \mu \left( s\right) \right\vert =q_{s}$ for all $i\in I$ and all $%
s\in S$. A matching $\mu $ \textbf{violates} the priority $\succ _{s}$\ of $%
i\notin \mu \left( s\right) $ over $j\in \mu \left( s\right) $ if $sP_{i}\mu
\left( i\right) \ $and $\left( i,j\right) \in \succ _{s}$. If a matching $%
\mu $ does not violate the priority $\succ _{s}$ for any $s\in S$, then it
is said to be \textbf{fair }for $\succ $. A matching $\mu $ is \textbf{stable%
} for $\succ $ if it is individually rational, non-wasteful and fair for $%
\succ $.\footnote{%
Although the formally correct notation would be \textquotedblleft fair for $%
G $\textquotedblright\ and \textquotedblleft stable for $G$%
\textquotedblright , we write \textquotedblleft fair for $\succ $%
\textquotedblright\ and \textquotedblleft stable for $\succ $%
\textquotedblright\ as a shorthand, since all components of $G$ other than $%
\succ $ are fixed.}

A matching $\mu $ \textbf{is} \textbf{Pareto dominated} \textbf{by} $\mu
^{\prime }$ if $\mu ^{\prime }\left( i\right) P_{i}\mu \left( i\right) $\
for some $i\in I$ and $\mu ^{\prime }\left( i^{\prime }\right) R_{i^{\prime
}}\mu \left( i^{\prime }\right) $ for all $i^{\prime }\in I$. Moreover, a
matching $\mu $ \textbf{is weakly} \textbf{Pareto dominated} \textbf{by} $%
\mu ^{\prime }$ if it is Pareto dominated by $\mu ^{\prime }$ or $\mu \left(
i\right) =\mu ^{\prime }\left( i\right) $ for all $i\in I$. A matching $\mu $
is \textbf{constrained efficient} for $\succ $\textbf{\ }if it is stable for 
$\succ $ and is not Pareto dominated by any stable matching for $\succ $.
Let $\mathcal{M}^{\ast }\left( \succ \right) $ be the set of all possible
constrained efficient matchings for $\succ $.

By definition of constrained efficiency, efficiency is considered only
within the set of matchings that are stable for $\succ $. Hence, the
priority structure $\succ $ constrains the set of matchings among which
efficiency is assessed. The weaker $\succ $ is at a school $s$, in the sense
that it imposes fewer priority comparisons, the less restrictive this
constraint becomes and the easier it is to attain efficiency. Conversely,
imposing unnecessary refinements such as arbitrary tie-breaking makes the
priority structure more restrictive and may rule out Pareto-superior stable
matchings.

\subsection{Priority orders}

We define the following properties of binary relations. A priority relation
for $s$ denoted by $\succ _{s}$ is

\begin{description}
\item \textbf{irreflexive} if $\left( i,i\right) \notin \succ _{s}$ for all $%
i\in I,$

\item \textbf{asymmetric} if $\left( i,j\right) \in \succ _{s}$ implies $%
\left( j,i\right) \notin \succ _{s}$, for all $i,j\in I,$

\item \textbf{total}\textit{\ }if\textit{\ }$i\neq j$ implies\textit{\ }$%
\left( i,j\right) \in \succ _{s}$ or $\left( j,i\right) \in \succ _{s}$ for
all $i,j\in I$,

\item \textbf{transitive} if $\left( i,j\right) \in \succ _{s}$ and $\left(
j,k\right) \in \succ _{s}$ imply $\left( i,k\right) \in \succ _{s}$, for all 
$i,j,k\in I,$

\item \textbf{negatively} \textbf{transitive}\textit{\ }if $\left(
i,j\right) \notin \succ _{s}$ and $\left( j,k\right) \notin \succ _{s}$
imply $\left( i,k\right) \notin \succ _{s}$, for all $i,j,k\in I$,

\item \textbf{acyclic }if for all $K\in \left\{ 1,2,\cdots \right\} $ and
for all $i_{0},i_{1},\cdots ,i_{K}\in I$, $\left( i_{k-1},i_{k}\right) \in
\succ _{s}$ and $\left( i_{k},i_{k-1}\right) \notin \succ _{s}$ for all $%
k\in \left\{ 1,\cdots ,K\right\} $ imply $\left( i_{K},i_{0}\right) \notin
\succ _{s}$.
\end{description}

Following Duggan (1999), we consider the following specific binary relations.

\begin{itemize}
\item A strict \textbf{linear order} is defined as an asymmetric,
transitive, and total binary relation.

\item A strict \textbf{weak order} is defined as a negatively transitive and
asymmetric binary relation.

\item A strict \textbf{partial order} is defined as a transitive and
asymmetric binary relation.
\end{itemize}

Throughout the paper, we consider only asymmetric binary priority relations.
In other words, we impose the requirement that, for any school $s$, if
student $i$ has priority over student $j$ at $s$, then student $j$ cannot
have priority over student $i$ at $s$. Moreover, in the remainder of the
paper we omit the qualifier \textquotedblleft strict\textquotedblright\ and
write, for instance, \textquotedblleft partial order\textquotedblright\ to
mean \textquotedblleft strict partial order.\textquotedblright \footnote{%
In the standard economics literature, preference or priority relations are
usually taken to be reflexive. The corresponding reflexive relations can be
obtained from the above irreflexive relations in a straightforward way.
Given a strict weak order $\succ _{s}$, define a reflexive relation $%
\succsim _{s}$ by 
\begin{equation*}
\left( i,j\right) \in \succsim _{s}\Leftrightarrow \left( j,i\right) \notin
\succ _{s}
\end{equation*}%
Then, $\succsim _{s}$ is a weak order; that is, a complete and transitive
binary relation. Similarly, if $\succ _{s}$ is a strict linear order, then
the relation $\succsim _{s}$ defined in the same way is a linear order; that
is, a complete, transitive, and antisymmetric binary relation. Conversely,
given a weak order $\succsim _{s}$, its asymmetric part $\succ _{s}$,
defined by 
\begin{equation*}
\left( i,j\right) \in \succ _{s}\Leftrightarrow \left( i,j\right) \in
\succsim _{s}\text{ and }\left( j,i\right) \notin \succsim _{s}
\end{equation*}%
is a strict weak order. Likewise, given a linear order $R$, its asymmetric
part $P$ is a strict linear order.}

Let $\mathcal{A}\subseteq \mathcal{B}$ be the set of all asymmetric and
acyclic relations on $I$. Let $\mathcal{P}\subseteq \mathcal{B}$ be the set
of all strict partial orders on $I$. Let $\mathcal{W}\subseteq \mathcal{B}$
be the set of all strict weak orders on $I$. Finally, let $\mathcal{L}%
\subseteq \mathcal{B}$ be the set of all strict linear priority orders on $I$%
. This gives the following standard hierarchy: 
\begin{equation*}
\mathcal{L\subseteq W\subseteq P\subseteq A\subseteq B}.
\end{equation*}%
See, for example, Duggan (1999) on this.

As noted in Section 1, almost all previous studies on school choice
represent schools' priority relations as weak orders. In Section 2.2, we
discuss why it is meaningful to consider priority relations that are not
weak orders.

\subsection{Beyond weak priority orders and partial orders}

As stated earlier, almost all previous studies on school choice assume that
the priority relation of every school is a weak order. On the other hand,
Che et al. (2019a, b), Dur et al. (2019), Kitahara and Okumura (2020, 2021,
2024), and Kuvalekar (2023), discuss matching models where the preference or
priority relations of some agents are not weak orders.

\subsubsection{Ignoring Small Differences}

One reason to move beyond weak priority orders is that, in practice, schools
or matching authorities may wish to ignore small differences in evaluation
scores. The following example is due to Kitahara and Okumura (2021). Suppose
that a school's priority is based on a test scored out of 100 points. If the
score difference between two students is at least 10 points, then the
student with the higher score is given priority. If the difference is less
than 10 points, then the rule treats the two students as incomparable and
does not assign priority between them.

Such a priority relation is generally not a weak order. To see this,
consider three students $i$, $j$, and $k$ with scores 100, 92, and 84, in
the entrance exams at school $s$, respectively. Let $\succ $ be induced by
the rule above. Then $\left( i,j\right) \notin \succ _{s}$ and $\left(
j,k\right) \notin \succ _{s}$, because both score differences are less than
10. However, $\left( i,k\right) \in \succ _{s}$, since the score difference
between them is at least 10. Hence, negative transitivity fails. Therefore,
the priority relation is not a weak order.

However, this priority relation based on the rule is transitive. Indeed, if
student $i$ has priority over student $j$, and $j$ has priority over $k$,
then the score difference between $i$ and $k$ must be at least 20, so $i$
also has priority over $k$. Therefore, the relation can naturally be
represented as a strict partial order.

\subsubsection{Multiple Priorities}

There may be several distinct priority rules, and these rules need not
always be compatible with one another. In the old school choice system in
the city of Boston, for example, sibling priority and walk-zone priority
coexisted (Abdulkadiro\u{g}lu et al., 2005, 2006). Conflicts may arise
between the priority ordering based on sibling status and that based on
walk-zone status, particularly when one student has only sibling priority
while another has only walk-zone priority. Such conflicts may prevent the
existence of a nonwasteful matching that simultaneously respects both
priority orderings. Therefore, in such school choice systems, assignments
are typically determined by first specifying which priority rule is more
important and then using a single linear priority order for each school.

Here, we consider an alternative approach under which each priority rule is
respected as long as it is compatible with the others, whereas, whenever two
priority rules conflict, neither of them is applied. As a result, the
resulting priority relation need not be a weak order. For example, suppose
there are three students $i$, $j$, and $k$ and there are two priority orders
at school $s$, where 
\begin{equation*}
\succ _{s}^{1}:\text{ }i,j,k,\text{ }\succ _{s}^{2}:\text{ }j,k,i.
\end{equation*}%
Then, applying the rule above, 
\begin{equation*}
\succ _{s}=\left\{ \left( j,k\right) \right\} ,
\end{equation*}%
which does not satisfy negative transitivity. This is because, $\left(
j,i\right) \notin \succ _{s}$ and $\left( i,k\right) \notin \succ _{s},$ but 
$\left( j,k\right) \in \succ _{s}$.

By contrast, when each underlying priority relation is a weak order, the
combined priority relation is a strict partial order. A formal proof of this
result is provided in Section 2 of Kitahara and Okumura (2023), an earlier
manuscript from which the present paper evolved.

\subsubsection{Consenting}

Following Kesten (2010) and Dur et al. (2019), we consider a setting in
which each student may selectively consent to priority violations; that is,
each student can choose at which school and in favor of whom to consent to a
priority violation.

Suppose that $\succ _{s}$ is a linear order for all $s\in S$. For each $s\in
S,$ let $C:S\twoheadrightarrow I\times I$ be the correspondence, defined in
Dur et al. (2019) where $(i,j)\in C\left( s\right) $ means that student $i$
agrees not to object to $j$'s assignment to $s$, even if $i$ has higher
priority than $j$ at $s$.

Dur et al. (2019) introduce a weakened notion of stability for $\left( \succ
,C\right) =\left( \left( \succ _{s}\right) _{s\in S},\left( C_{s}\right)
_{s\in S}\right) $, called partial stability: a matching $\mu $ is \textbf{%
partially stable} if for every $(i,j,s)$ satisfying $j\in \mu \left(
s\right) ,$ $sP_{i}\mu \left( i\right) \ $and $\left( i,j\right) \in \succ
_{s},$ it holds $(i,j)\in C\left( s\right) $.

Dur et al. (2019) impose the following monotonicity condition: if $(i,j)\in
C\left( s\right) $, then for every $k$ such that $i\succ _{s}k\succ _{s}j$,
we also have $(i,k)\in C\left( s\right) $. For example, suppose that
priorities at a school are determined by test scores. Then, if student $i$
consents to a priority violation in favor of student $j$ at school $s$,
student $i$ must also consent to a priority violation in favor of any
student $j^{\prime }$ whose score is higher than $j$'s but lower than $i$'s.

Kitahara and Okumura (2021) define that for a given $\left( \succ
_{s},C\left( s\right) \right) $, let $\succ _{s}^{C}$ be a binary relation
on $I$ satisfying $i\succ _{s}^{C}j$ if and only if $i\succ _{s}j$ and $%
(i,j)\notin C\left( s\right) $. Then, they show that for any $\left( \succ
,C\right) $, a matching is partially stable for $\left( \succ ,C\right) $ if
and only if it is stable for $\succ ^{C}=\left( \succ _{s}^{C}\right) _{s\in
S}$.

Moreover, Kitahara and Okumura (2021) show that $\succ _{s}^{C}$ is a
partial order if a condition which is weaker than the monotonicity condition
introduced by Dur et al. (2019) is satisfied. However, if the condition is
not satisfied, then $\succ _{s}^{C}$ may not be transitive. For example, we
consider the following case.

Suppose $i\succ _{s}k\succ _{s}j$, $\left( i,j\right) \in C\left( s\right) ,$
$\left( i,k\right) \notin C\left( s\right) $ and $\left( k,j\right) \notin
C\left( s\right) $. Then, $\left( i,k\right) \in \succ _{s}^{C}$ and $\left(
k,j\right) \in \succ _{s}^{C}$ but $\left( i,j\right) \notin \succ _{s}^{C}$%
. Therefore, $\succ _{s}^{C}$ is not transitive. For example, $i$ and $k$
are majority students and $j$ is a minority student. Student $i$ consents to
a priority violation in favor of student $j$ but does not consent to one in
favor of student $k,$ because $k$ is a majority student and $j$ is a
minority student. On the other hand, student $k$ does not consent to any
priority violation.

In the general setting, although $\succ _{s}^{C}$ satisfies acyclicity, it
may not be a partial order.

\subsection{Linear order extension and existence result}

For $\succ _{s}\in \mathcal{B},$ let $\succ _{s}^{\ast }$ be a \textbf{%
linear order} \textbf{extension} of $\succ _{s}$ if $\succ _{s}^{\ast }\in 
\mathcal{L}$ and $\succ _{s}\subseteq \succ _{s}^{\ast }$. Note that there
may exist multiple linear order extensions for $\succ _{s}$. Moreover, we
let $E\left( \succ _{s}\right) $ be the set of all linear order extensions
of $\succ _{s}$. On the existence of linear order extensions of $\succ _{s}$%
, the following result is known.

\begin{remark}
For $\succ _{s}\in \mathcal{B},$ $E\left( \succ _{s}\right) \neq \emptyset $
if and only if $\succ _{s}\in \mathcal{A}$.
\end{remark}

Since any partial order is acyclic and asymmetric, for any $\succ _{s}\in 
\mathcal{P},$ there is a linear order extension of $\succ _{s}$. Let $\succ
^{\ast }$ be a \textbf{linear order} \textbf{extension profile }of $\succ $
if $\succ _{s}^{\ast }\in E\left( \succ _{s}\right) $ for all $s\in S$.
Moreover, we let $\mathcal{E}\left( \succ \right) $ be the set of all
possible linear order extension profiles of $\succ $.

The following result is known (see, e.g., Kitahara and Okumura (2024)).

\begin{remark}
If $\succ ^{\prime }\in \mathcal{E}\left( \succ \right) $ and $\mu $ is
stable for $\succ ^{\prime }$, then $\mu $ is also stable for $\succ $.
\end{remark}

By this result, if there is a linear order extension profile of $\succ $,
then a stable matching for the linear order extension profile is also stable
for $\succ $. The following result is also due to Kitahara and Okumura
(2024, Lemma 1).

\begin{remark}
If $\succ _{s}\in \mathcal{A}$ for all $s\in S$, then there is a stable
matching for $\succ $.
\end{remark}

By Remark 1, if $\succ _{s}\in \mathcal{A}$, then $E\left( \succ _{s}\right)
\neq \emptyset $. By Remark 2, a stable matching for the extension profile
is also stable for $\succ $. Since Gale and Shapley (1962) show that there
is a stable matching for a profile of linear priority orders, we have Remark
3.

\section{EADA mechanism}

We introduce the (simplified) \textbf{EADA} mechanism for $\succ $, which is
introduced by Kesten (2010) and modified by Tang and Yu (2014).\footnote{%
Regarding the modification by Tang and Yu (2014), see footnote 7 of Kitahara
and Okumura (2024).} The following mechanism is almost the same as that
introduced by Tang and Yu (2014, Subsection 4.2), but they only consider
weak priority orders. That is, we show that the mechanism also attains a
constrained efficient matching for $\succ $ even when $\succ _{s}$ is not a
weak order but a partial order for some $s\in S$.

Let 
\begin{equation*}
\hat{G}=\left( \hat{I},\hat{S},\hat{P},\hat{\succ},\hat{q}\right)
\end{equation*}%
be a reduced school choice problem in which $\hat{I}\subseteq I,$ $\hat{S}%
\subseteq S,$ and $\hat{\succ}_{s}$ is a linear order for all $s\in S$.

Let $DA\left( \hat{G}\right) $ be the result of the Gale and Shapley's
(1962) SPDA algorithm for $\hat{G}$. A school $s\in \hat{S}$ is said to be 
\textbf{underdemanded} at a matching $\mu $ with $\hat{G}$ if $\mu \left(
i\right) \hat{R}_{i}s$ for all $i\in \hat{I}$. Then, we have the following
result.

\begin{remark}
A school $s$ is underdemanded at $DA\left( \hat{G}\right) $ with $\hat{G}$\
if and only if $s$ never rejects at any step of the SPDA algorithm for $\hat{%
G}$.
\end{remark}

Next, we introduce the EADA mechanisms for $G$.\footnote{%
Although the original EADA mechanism is introduced for the case in which $%
\succ _{s}=\emptyset $ for every $s$, and mechanisms outside this case have
sometimes been called \textquotedblleft variants of EADA
mechanism\textquotedblright , this case is included in our framework.
Indeed, the empty relation is a weak order, so $\succ _{s}=\emptyset $ for
every $s$ implies $\succ \in \mathcal{W}^{\left\vert S\right\vert }$. Hence,
the original EADA mechanism is a special case of the mechanism for $\succ
\in \mathcal{W}^{\left\vert S\right\vert }$.}

\begin{description}
\item[Round $0$] Choose $\succ ^{\ast }\in \mathcal{E}\left( \succ \right) $.

\item[Round $1$] Let $I^{1}=I,$ $S^{1}=S,$ and $P^{1}=P$. Run the SPDA for $%
G^{1}=\left( I^{1},S^{1},P^{1},\succ ^{\ast },q\right) $. Let $U^{1}$ be the
set of underdemanded schools at $DA\left( G^{1}\right) $; that is, each $%
s\in U^{1}$ never rejects any student throughout the SPDA in this Round.
Moreover, let 
\begin{equation*}
E^{1}=\bigcup\nolimits_{s\in U^{1}\cup \left\{ \emptyset \right\} }DA\left(
G^{1}\right) \left( s\right) ,
\end{equation*}%
which is the set of students who are matched to an underdemanded school (or
unmatched) at $DA\left( G^{1}\right) $. Let $\hat{\mu}^{1}=DA\left(
G^{1}\right) $.

\item[Round $k$] Let $I^{k}=I^{k-1}\setminus E^{k-1}$ and $%
S^{k}=S^{k-1}\setminus U^{k-1}$. For all $i\in I^{k},$ let 
\begin{equation*}
Z_{i}^{k}=\left\{ s\in S^{k}\text{ }\left\vert \text{ }sP_{i}^{k-1}\emptyset
,\text{ }\left( j,i\right) \in \succ _{s}\text{ and }sP_{j}DA\left(
G^{k-1}\right) (j),\text{ for some }j\in E^{k-1}\right. \right\} \text{,}
\end{equation*}%
and let $P_{i}^{k}$ be a linear order such that for all $s\in Z_{i}^{k}$, $%
\emptyset P_{i}^{k}s$ and for all $s^{\prime },s^{\prime \prime }\in \left(
S\cup \left\{ \emptyset \right\} \right) \setminus Z_{i}^{k},$ $s^{\prime
}P_{i}^{k-1}s^{\prime \prime }$ implies $s^{\prime }P_{i}^{k}s^{\prime
\prime }$. Run the SPDA for $G^{k}=\left( I^{k},S^{k},P^{k},\succ ^{\ast
},q\right) $, where $\succ ^{\ast }$ is restricted to $I^{k}$. Let $U^{k}$
be the set of underdemanded schools at $DA\left( G^{k}\right) $; that is,
each $s\in U^{k}$ never rejects any student throughout the SPDA of this
Round. Moreover, let 
\begin{equation*}
E^{k}=\bigcup\nolimits_{s\in U^{k}\cup \left\{ \emptyset \right\} }DA\left(
G^{k}\right) \left( s\right) ,
\end{equation*}%
which is the set of students who are matched to an underdemanded school (or
unmatched) at $DA\left( G^{k}\right) $.\footnote{%
In each Round $k,$ $E^{k}$ is not empty. First, if no student is rejected,
then all schools are underdemanded and $E^{k}=I^{k}$. Otherwise, then we can
let $i$ be the student who is rejected last by a school. Then, $i$ must be
in $E^{k}$.} Let $\hat{\mu}^{k}$ be such that 
\begin{eqnarray*}
\hat{\mu}^{k}\left( i\right) &=&DA\left( G^{k}\right) \left( i\right) \text{
for all }i\in I^{k}, \\
\hat{\mu}^{k}\left( i\right) &=&DA\left( G^{\kappa }\right) \left( i\right) 
\text{ for all }i\in E^{\kappa }\text{ and all }\kappa =1,\cdots ,k-1.
\end{eqnarray*}
\end{description}

The procedure terminates after Round $K$ when $I^{K}=E^{K}$, so that all
remaining students are eliminated in that round and 
\begin{equation*}
EA\left( \succ ,\succ ^{\ast }\right) =\hat{\mu}^{K}.
\end{equation*}

In this mechanism, after running DA in each round, the assignments of
students who are matched to underdemanded schools, as well as those of
students who remain unmatched, are fixed. Consider Round $k$. To ensure that
no student whose assignment has already been fixed in an earlier round has
justified envy toward a student $i$ who remains active, the mechanism
modifies $i$'s preference so that $i$ no longer applies to any school in $%
Z_{i}^{k}$. Specifically, every school in $Z_{i}^{k}$ is made unacceptable
to $i$ by moving it below the outside option in $i$'s preference ordering $%
P_{i}^{k}$.

We have the following result.

\begin{theorem}
If $\succ \in \mathcal{P}^{\left\vert S\right\vert }$ and $\succ ^{\ast }\in 
\mathcal{E}\left( \succ \right) $, then $EA\left( \succ ,\succ ^{\ast
}\right) \in \mathcal{M}^{\ast }\left( \succ \right) $.
\end{theorem}

The proof is provided in the Appendix.

By Theorem 1, for any $\succ \in \mathcal{P}^{\left\vert S\right\vert }$, $%
EA\left( \succ ,\succ ^{\ast }\right) $ is constrained efficient for $\succ $%
.

\begin{remark}
(Kitahara and Okumura, 2024) Fix any $\succ \in \mathcal{P}^{\left\vert
S\right\vert }$. If $\mu \in \mathcal{M}^{\ast }\left( \succ \right) $, then
there is $\succ ^{\ast }\in \mathcal{E}\left( \succ \right) $ such that $\mu
=EA\left( \succ ,\succ ^{\ast }\right) $.
\end{remark}

\section{SIC mechanism}

Next, we consider another mechanism. Let 
\begin{equation*}
D_{i}\left( \mu \right) =\left\{ \left. j\in I\text{ }\right\vert \text{ }%
\mu \left( i\right) P_{j}\mu \left( j\right) \right\}
\end{equation*}%
be the set of students who prefer $\mu \left( i\right) $ to their matched
school. Next, let 
\begin{equation*}
Y_{i}\left( \mu ,\succ \right) =\left\{ \left. i^{\prime }\in D_{i}\left(
\mu \right) \text{ }\right\vert \text{ }\left( j,i^{\prime }\right) \in
\succ _{\mu \left( i\right) }\text{ for some }j\in D_{i}\left( \mu \right)
\right\}
\end{equation*}%
be the set of students in $D_{i}\left( \mu \right) $ whose priority is lower
than some other student in $D_{i}\left( \mu \right) $.

Finally, let $X_{i}\left( \mu ,\succ \right) =D_{i}\left( \mu \right)
\setminus Y_{i}\left( \mu ,\succ \right) $; that is, the set of students
who\ prefer $\mu \left( i\right) $ to their assignment and whose priorities
for $\mu \left( i\right) $ are not lower than any student in $D_{i}\left(
\mu \right) $. Of course, $X_{i}\left( \mu \right) $ may not be singleton,
because $\succ $ is allowed not to be total.

Let $\phi =\left\{ i_{1}i_{2},i_{2}i_{3},\cdots ,i_{q}i_{q+1}\right\} $ be a
set of ordered pairs of $I$. The set $\phi =\left\{
i_{1}i_{2},i_{2}i_{3},\cdots ,i_{q}i_{q+1}\right\} $ ($q\geq 2$) is called a 
\textit{cycle} if $i_{1},\cdots ,i_{q}$ ($q\geq 2$) are distinct and $%
i_{q+1}=i_{1}$.

\begin{definition}
A stable improvement cycle of $\mu $ for $\succ $ denoted by $\phi =\left\{
i_{1}i_{2},\cdots ,i_{q}i_{q+1}\right\} $ is a cycle such that: (i) $\mu
\left( i_{l}\right) \in S$, and (ii) $i_{l}\in X_{i_{l+1}}\left( \mu ,\succ
\right) $ for any $l=1,\cdots ,q$.
\end{definition}

Given a matching $\mu $ and a stable improvement cycle of $\mu $ denoted by $%
\phi =\left\{ i_{1}i_{2},\cdots ,i_{q}i_{q+1}\right\} $, we define a new
matching $\phi \circ \mu $ by 
\begin{equation*}
\left( \phi \circ \mu \right) \left( j\right) = 
\begin{cases}
\mu \left(j\right) \text{ if }j\notin \left\{ i_{1},\cdots ,i_{q}\right\} \\ 
\mu \left(i_{l+1}\right) \text{ if }j=i_{l}%
\end{cases}%
\end{equation*}

Now, we introduce the SIC (stable improvement cycle) mechanism.

\begin{description}
\item[Round $0$] Choose $\succ ^{\ast }\in \mathcal{E}\left( \succ \right) $.

\item[Round $1$] Run the SPDA for $G^{\ast }=\left( I,S,P,\succ ^{\ast
},q\right) $ and let $\mu ^{1}=DA\left( G^{\ast }\right) $.

\item[Round $t\geq 2$] Let $E^{t}$ be a set of ordered pairs of $I$ such
that $ij\in E^{t}$ if and only if $i\in X_{j}\left( \mu ^{t-1},\succ \right) 
$. If there is a cycle in $E^{t}$, then choose one denoted by $\phi ^{t}$
and let $\mu ^{t}=\left( \phi ^{t}\circ \mu ^{t-1}\right) $. Otherwise, the
mechanism is completed and $\mu ^{t-1}$ is the resulting matching.
\end{description}

For a given sequence of cycle selections, let $\mu ^{T}$ denote the terminal
matching.

The result of this mechanism depends on (1) which extension profile is
chosen in Round $0$ and (2) which improvement cycle is selected in Rounds $%
2, $ $3,$ and all subsequent Rounds. Thus, we let $\mathcal{SI}\left( \succ
,\succ ^{\ast }\right) $ denote the \textit{set} of matchings that can be
reached by the SIC procedure from the DA outcome under $\succ ^{\ast },$
over all possible sequences of stable-improvement-cycle selections.

The following result is shown by Kitahara and Okumura (2021).

\begin{remark}
(Kitahara and Okumura, 2021) If $\succ \in \mathcal{P}^{\left\vert
S\right\vert }$, then any $\mu \in \mathcal{SI}\left( \succ ,\succ ^{\ast
}\right) $ is a constrained efficient matching for $\succ $.
\end{remark}

More precisely, Kitahara and Okumura (2021) take a stable matching as an
input and discuss how to Pareto improve it while preserving stability. In
our setting, the input to the mechanism is not a stable matching but rather
a linear order extension of $\succ $. However, by using this linear
extension to compute a stable matching in Step 1, our mechanism aligns with
theirs.

Moreover, due to Kitahara and Okumura (2021), we immediately have the
following result.

\begin{remark}
(Kitahara and Okumura, 2021) If $\mu $ is a constrained efficient matching
for $\succ \in \mathcal{P}^{\left\vert S\right\vert }$, then there is $\succ
^{\ast }\in \mathcal{E}\left( \succ \right) $ such that $\mathcal{SI}\left(
\succ ,\succ ^{\ast }\right) $ includes $\mu $.
\end{remark}

Kitahara and Okumura (2024) point out that the two mechanisms are outcome
equivalent. Formally, we have the following result.

\begin{remark}
(Kitahara and Okumura, 2024) For $\succ \in \mathcal{P}^{\left\vert
S\right\vert }$, 
\begin{equation}
\left\{ EA\left( \succ ,\succ ^{\ast }\right) \right\} _{\succ ^{\ast }\in 
\mathcal{E}\left( \succ \right) }=\bigcup\limits_{\succ ^{\ast }\in 
\mathcal{E}\left( \succ \right) }\mathcal{SI}\left( \succ ,\succ ^{\ast
}\right) =\mathcal{M}^{\ast }\left( \succ \right) .  \notag
\end{equation}
\end{remark}

Kitahara and Okumura (2021) show that, for $\succ \in \mathcal{P}%
^{\left\vert S\right\vert }$, 
\begin{equation*}
\bigcup\limits_{\succ ^{\ast }\in \mathcal{E}\left( \succ \right) }\mathcal{%
SI}\left( \succ ,\succ ^{\ast }\right) =\mathcal{M}^{\ast }\left( \succ
\right) .
\end{equation*}%
Kitahara and Okumura (2024) show that, for $\succ \in \mathcal{P}%
^{\left\vert S\right\vert }$, 
\begin{equation}
\left\{ EA\left( \succ ,\succ ^{\ast }\right) \right\} _{\succ ^{\ast }\in 
\mathcal{E}\left( \succ \right) }=\mathcal{M}^{\ast }\left( \succ \right) , 
\notag
\end{equation}%
and therefore, they point out the equivalence result.

Note that it need not be the case that 
\begin{equation*}
\mathcal{SI}\left( \succ ,\succ ^{\ast }\right) =\left\{ EA\left( \succ
,\succ ^{\ast }\right) \right\} .
\end{equation*}%
This is because the outcome of the SIC mechanism also depends on which
stable improvement cycle is selected at each step, and hence $\mathcal{SI}%
\left( \succ ,\succ ^{\ast }\right) $ may contain multiple matchings.
Nevertheless, both mechanisms can attain every constrained efficient
matching under $\succ $.

\section{Acyclic priorities}

Thus far, we show that the two mechanisms have the same range, as long as $%
\succ \in \mathcal{P}^{\left\vert S\right\vert }$. By Remark 1, both of them
are well-defined if $\succ \in \mathcal{A}^{\left\vert S\right\vert }$.
However, on this broader domain, the two mechanisms need not be outcome
equivalent. Moreover, each mechanism may fail to attain a constrained
efficient matching, albeit in a different way.

First, the following example shows that the EADA mechanism may not attain
any constrained efficient matching.

\subsubsection*{\textbf{Example 1}}

Let $I=\left\{ i_{1},\cdots ,i_{5}\right\} $ and $S=\left\{ s_{1},\cdots
,s_{4}\right\} $, $q_{s}=1$ for all $s\in S$. Let the students' preferences $%
P$ and the schools' priority orders $\succ ^{\ast }$ be as follows:

\begin{center}
$%
\begin{tabular}{lllllllll}
$P_{i_{1}}$ & $P_{i_{2}}$ & $P_{i_{3}}$ & $P_{i_{4}}$ & $P_{i_{5}}$ & $\succ
_{s_{1}}^{\ast }$ & $\succ _{s_{2}}^{\ast }$ & $\succ _{s_{3}}^{\ast }$ & $%
\succ _{s_{4}}^{\ast }$ \\ \hline
$s_{2}$ & $s_{1}$ & $s_{1}$ & $s_{3}$ & $s_{1}$ & $i_{4}$ & $i_{2}$ & $i_{3}$
& $i_{4}$ \\ 
$s_{1}$ & $s_{2}$ & $s_{4}$ & $s_{4}$ & $s_{4}$ & $i_{1}$ & $i_{1}$ & $i_{4}$
& $i_{5}$ \\ 
$\emptyset $ & $\emptyset $ & $s_{3}$ & $\emptyset $ & $\emptyset $ & $i_{5}$
& $i_{3}$ & $i_{2}$ & $i_{3}$ \\ 
&  & $\emptyset $ &  &  & $i_{3}$ & $i_{4}$ & $i_{1}$ & $i_{1}$ \\ 
&  &  &  &  & $i_{2}$ & $i_{5}$ & $i_{5}$ & $i_{2}$%
\end{tabular}%
$
\end{center}

Let $\succ $ be such that%
\begin{eqnarray*}
&\succ &_{s_{1}}=\succ _{s_{1}}^{\ast }\setminus \left\{ \left(
i_{5},i_{2}\right) \right\} , \\
&\succ &_{s_{i}}=\succ _{s_{i}}^{\ast }\text{ for all }i=2,3, \\
&\succ &_{s_{4}}=\succ _{s_{4}}^{\ast }\setminus \left\{ \left(
i_{5},i_{3}\right) \right\} .
\end{eqnarray*}%
Since $\left( i_{5},i_{3}\right) ,\left( i_{3},i_{2}\right) \in \succ
_{s_{1}}$ but $\left( i_{5},i_{2}\right) \notin \succ _{s_{1}}$, $\succ
_{s_{1}}$ is not transitive and thus not a partial order. However, $\succ
_{s_{1}}$ is acyclic, because $\left( i_{2},i_{5}\right) \notin \succ
_{s_{1}}$. Note that $\succ ^{\ast }\in \mathcal{E}\left( \succ \right) $.

In this example, there are only two stable matchings $\mu $ and $\mu
^{\prime }$ for $\succ $ such that 
\begin{eqnarray*}
\left( \mu \left( i_{1}\right) ,\cdots ,\mu \left( i_{5}\right) \right)
&=&\left( s_{1},s_{2},s_{3},s_{4},\emptyset \right) , \\
\left( \mu ^{\prime }\left( i_{1}\right) ,\cdots ,\mu ^{\prime }\left(
i_{5}\right) \right) &=&\left( s_{1},s_{2},s_{4},s_{3},\emptyset \right) .
\end{eqnarray*}%
Since $\mu ^{\prime }$ Pareto dominates $\mu $, $\mu ^{\prime }$ is uniquely
constrained efficient for $\succ $.

We let $\succ ^{\ast \ast }\in \mathcal{E}\left( \succ \right) \setminus
\left\{ \succ ^{\ast }\right\} $ be such that $\succ _{s_{j}}^{\ast }=\succ
_{s_{j}}^{\ast \ast }=\succ _{s_{j}}$ for all $j=1,2,3$ and 
\begin{equation*}
\succ _{s_{4}}^{\ast \ast }:i_{4}\text{ }i_{3}\text{ }i_{5}\text{ }i_{1}%
\text{ }i_{2}.
\end{equation*}%
At $s_{1},s_{2}$ and $s_{3}$, the acyclic priority relation has a unique
linear-order extension. At $s_{4}$, the only undetermined comparison is
between $i_{3}$ and $i_{5}$. Hence, $\succ $ has exactly two linear-order
extension profiles $\succ ^{\ast }$ and $\succ ^{\ast \ast }$.

We derive $EA\left( \succ ,\succ ^{\ast }\right) $. In Round 1 
\begin{equation*}
DA\left( G^{1}\right) =\mu .
\end{equation*}%
Then, $E^{1}=\left\{ i_{5}\right\} $; that is, $i_{5}$ is eliminated,
because their assignment is $\emptyset $. Then, since $\left(
i_{5},i_{3}\right) \in \succ _{s_{1}}$, $Z_{i_{3}}^{2}=\left\{ s_{1}\right\} 
$. Hence $\emptyset P_{i_{3}}^{2}s_{1}$, to prevent $i_{3}$ being matched to 
$s_{1}$.

Second, 
\begin{equation*}
\left( DA\left( G^{2}\right) \left( i_{1}\right) ,\cdots ,DA\left(
G^{2}\right) \left( i_{4}\right) \right) =\left(
s_{2},s_{1},s_{4},s_{3}\right) ,
\end{equation*}%
and the EADA terminates. Therefore,%
\begin{equation*}
\left( EA\left( \succ ,\succ ^{\ast }\right) \left( i_{1}\right) ,\cdots
,EA\left( \succ ,\succ ^{\ast }\right) \left( i_{5}\right) \right) =\left(
s_{2},s_{1},s_{4},s_{3},\emptyset \right) \text{.}
\end{equation*}%
However, $EA\left( \succ ,\succ ^{\ast }\right) $ is not fair for $\succ ,$
because $\left( i_{3},i_{2}\right) \in \succ _{s_{1}}$and $%
s_{1}P_{i_{3}}s_{4}$.

The failure is driven by the non-transitivity of $\succ _{s_{1}}$: although $%
\left( i_{5},i_{3}\right) \in \succ _{s_{1}}$ and $\left( i_{3},i_{2}\right)
\in \succ _{s_{1}}$, $\left( i_{5},i_{2}\right) \notin \succ _{s_{1}}$.
Since $\left( i_{5},i_{3}\right) \in \succ _{s_{1}}$, the elimination of $%
i_{5}$ causes $s_{1}$ to be removed from $i_{3}$'s acceptable choice. This
allows $i_{2}$ to be assigned to $s_{1},$ even though $\left(
i_{3},i_{2}\right) \in \succ _{s_{1}}$ and $s_{1}P_{i_{3}}s_{4}$. If $\succ
_{s_{1}}$ were transitive while retaining $\left( i_{5},i_{3}\right) \in
\succ _{s_{1}}$ and $\left( i_{3},i_{2}\right) \in \succ _{s_{1}}$, $\left(
i_{5},i_{2}\right) \in \succ _{s_{1}}$ would also hold, and the elimination
of $i_{5}$ would prevent $i_{2}$, as well as $i_{3}$, from applying to $%
s_{1} $.\footnote{%
If, instead, $\left( i_{5},i_{3}\right) \notin \succ _{s_{1}}$, then $%
s_{1}\notin Z_{i_{3}}^{2}$, so $_{3}$ applies to $s_{1}$ in Round 2. In that
case, both both $i_{2}$ and $i_{3}$ apply to $s_{1}$, and $i_{2}$ is
rejected because $i_{3}\succ _{s_{1}}^{\ast }i_{2}$. On the other hand, if $%
\left( i_{3},i_{2}\right) \notin \succ _{s_{1}}$, then the resulting
assignment does not violate the priority of $i_{3}$ over $i_{2}$ at $s_{1}$.}

Next, we derive $EA\left( \succ ,\succ ^{\ast \ast }\right) $. In Round 1, 
\begin{equation*}
\left( DA\left( G^{1}\right) \left( i_{1}\right) ,\cdots ,DA\left(
G^{1}\right) \left( i_{5}\right) \right) =\left(
s_{1},s_{2},s_{4},s_{3},\emptyset \right) .
\end{equation*}%
Then, since $s_{3}$ never rejects, $U^{1}=\left\{ s_{3}\right\} $ and $%
E^{1}=\left\{ i_{4},i_{5}\right\} $. Moreover, as in the previous case,$\
Z_{i_{3}}^{2}=\left\{ s_{1}\right\} $.

Second, 
\begin{equation*}
\left( DA\left( G^{2}\right) \left( i_{1}\right) ,\cdots ,DA\left(
G^{2}\right) \left( i_{3}\right) \right) =\left( s_{2},s_{1},s_{4}\right) .
\end{equation*}%
and the EADA terminates. Therefore,%
\begin{equation*}
\left( EA\left( \succ ,\succ ^{\ast \ast }\right) \left( i_{1}\right)
,\cdots ,EA\left( \succ ,\succ ^{\ast \ast }\right) \left( i_{5}\right)
\right) =\left( s_{2},s_{1},s_{4},s_{3},\emptyset \right) \text{.}
\end{equation*}%
However, $EA\left( \succ ,\succ ^{\ast \ast }\right) \left( =EA\left( \succ
,\succ ^{\ast }\right) \right) $ is also not fair for $\succ ,$ because $%
\left( i_{3},i_{2}\right) \in \succ _{s_{1}}$and $s_{1}P_{i_{3}}s_{4}$.

Hence, in this example, under any linear order extension profiles of $\succ
, $ the EADA mechanism fails to be fair for $\succ $. On the other hand, we
next show that, in this example, the SIC mechanism always attains a
constrained efficient matching.

First, we consider $\succ ^{\ast }\in \mathcal{E}\left( \succ \right) $ such
that the SPDA result for $\succ ^{\ast }$ is $\mu ^{\prime }$, as defined
above. Then, there must be no SIC. This is because, since $\left(
i_{3},i_{2}\right) \in \succ _{s_{1}}$, $i_{2}\in Y_{i_{1}}\left( \mu
^{\prime }\right) $. Thus, in this case, the SIC mechanism results in $\mu
^{\prime },$ which is uniquely constrained efficient for $\succ $.

Second, we consider another $\succ ^{\ast \ast }\in \mathcal{E}\left( \succ
\right) $ such that the SPDA result for $\succ ^{\ast \ast }$ is $\mu $.
Then, there is a unique SIC $\phi =\left\{ i_{3}i_{4},i_{4}i_{3}\right\} ,$
because $\left( i_{5},i_{3}\right) \notin \succ _{s_{4}}$. Thus, in any
case, the SIC mechanism results in the unique constrained efficient matching 
$\mu ^{\prime };$ that is, $\mathcal{SI}\left( \succ ,\succ ^{\prime
}\right) =\left\{ \mu ^{\prime }\right\} $ for all $\succ ^{\prime }\in 
\mathcal{E}\left( \succ \right) $.\newline

On the other hand, in the next example, if $\succ _{s}\in \mathcal{%
A\setminus P}$ for some $s$, any matching included in $\mathcal{SI}\left(
\succ ,\succ ^{\ast }\right) $ fails to attain any constrained efficient
matching for $\succ $. The example is due to Dur et al. (2019).

\subsubsection*{\textbf{Example 2}}

Let $I=\left\{ i_{1},\cdots ,i_{4}\right\} $ and $S=\left\{ s_{1},\cdots
,s_{4}\right\} $, $q_{s}=1$ for all $s\in S$. Let the students' preferences $%
P$ and the schools' priority orders $\succ ^{\ast }$ be as follows:

\begin{center}
$%
\begin{tabular}{llllllll}
$P_{i_{1}}$ & $P_{i_{2}}$ & $P_{i_{3}}$ & $P_{i_{4}}$ & $\succ
_{s_{1}}^{\ast }$ & $\succ _{s_{2}}^{\ast }$ & $\succ _{s_{3}}^{\ast }$ & $%
\succ _{s_{4}}^{\ast }$ \\ \hline
$s_{2}$ & $s_{3}$ & $s_{4}$ & $s_{2}$ & $i_{1}$ & $i_{2}$ & $i_{3}$ & $i_{4}$
\\ 
$s_{1}$ & $s_{2}$ & $s_{2}$ & $s_{4}$ & $i_{2}$ & $i_{1}$ & $i_{2}$ & $i_{3}$
\\ 
$\emptyset $ & $\emptyset $ & $s_{3}$ & $\emptyset $ & $i_{3}$ & $i_{3}$ & $%
i_{1}$ & $i_{1}$ \\ 
&  & $\emptyset $ &  & $i_{4}$ & $i_{4}$ & $i_{4}$ & $i_{2}$%
\end{tabular}%
$
\end{center}

Moreover, let $\succ $ be such that%
\begin{gather*}
\succ _{s_{i}}=\succ _{s_{i}}^{\ast }\text{ for all }i=1,3,4 \\
\succ _{s_{2}}=\succ _{s_{2}}^{\ast }\setminus \left\{ \left(
i_{1},i_{4}\right) \right\} \text{.}
\end{gather*}%
Since $\left( i_{1},i_{3}\right) ,\left( i_{3},i_{4}\right) \in \succ
_{s_{2}}$ but $\left( i_{1},i_{4}\right) \notin \succ _{s_{2}}$, $\succ
_{s_{2}}$ is not transitive and thus not a partial order. However, $\succ
_{s_{2}}\in \mathcal{A}$, because $\left( i_{4},i_{1}\right) \notin \succ
_{s_{2}}$. Moreover, $\succ ^{\ast }$ is the unique linear extension profile
of $\succ ;$ that is, $\mathcal{E}\left( \succ \right) =\left\{ \succ ^{\ast
}\right\} $.

Then, there are two stable matchings for $\succ $ denoted by $\mu $ and $\mu
^{\prime }$ where 
\begin{eqnarray*}
\left( \mu \left( i_{1}\right) ,\cdots ,\mu \left( i_{4}\right) \right)
&=&\left( s_{1},s_{2},s_{3},s_{4}\right) , \\
\left( \mu ^{\prime }\left( i_{1}\right) ,\cdots ,\mu ^{\prime }\left(
i_{4}\right) \right) &=&\left( s_{1},s_{3},s_{4},s_{2}\right) .
\end{eqnarray*}%
Since $\mu ^{\prime }$ Pareto dominates $\mu ,$ $\mu ^{\prime }$ is uniquely
constrained efficient for $\succ $. Further, the SPDA result with $\succ
^{\ast }$ is $\mu $.

Then, there is no stable improvement cycle of $\mu ^{1}=\mu $. This is
because, first, since $\left( i_{3},i_{4}\right) \in \succ _{s_{2}}$, $%
i_{4}\in Y_{i_{2}}\left( \mu \right) ,$ and $\left( i_{1},i_{3}\right) \in
\succ _{s_{2}}$, $i_{3}\in Y_{i_{2}}\left( \mu \right) $. Therefore, $%
\mathcal{SI}\left( \succ ,\succ ^{\ast }\right) =\left\{ \mu \right\} $ for $%
\succ ^{\ast },$ which is the unique linear order extension of $\succ $, but
it is Pareto dominated by the other stable matching $\mu ^{\prime }$.

To understand why the SIC mechanism fails to achieve the Pareto improvement,
observe that that both $i_{3}$ and $i_{4}$ prefer $s_{2}$ to their
assignments under $\mu $. Since $\left( i_{3},i_{4}\right) \in \succ
_{s_{2}},$ $i_{4}\in Y_{i_{2}}\left( \mu \right) $. Consequently, the SIC
mechanism cannot execute the Pareto-improving cycle that assigns $i_{2}$ to $%
s_{3}$, $i_{3}$ to $s_{4}$, and $i_{4}$ to $s_{2}$. In this case, this
restriction is unnecessarily strong because $i_{3}$, who prevents $i_{4}$
from receiving $s_{2}$, would themselves move to the more-preferred school $%
s_{4}$ as part of the same reassignment and would therefore no longer prefer 
$s_{2}$.

The non-transitivity of $\succ _{s_{2}}$ is also essential for this failure.
Although $\left( i_{1},i_{3}\right) \in \succ _{s_{2}}$ and $\left(
i_{3},i_{4}\right) \in \succ _{s_{2}}$, we have $\left( i_{1},i_{4}\right)
\notin \succ _{s_{2}}$. Thus, even though $s_{2}P_{i_{1}}s_{1}$, $i_{1}$
does not have justified envy toward $i_{4}$, who is assigned to $s_{2}$
under $\mu ^{\prime }$. If $\succ _{s_{2}}$ were transitive while retaining
the first two priority comparisons,, $\left( i_{1},i_{4}\right) \in \succ
_{s_{2}}$ would hold, and $\mu ^{\prime }$ would not be stable.

On the other hand, we consider the EADA mechanism instead of the SIC
mechanism. We must use $\succ ^{\ast }$ as the unique extension of $\succ $.
In Round 1, 
\begin{equation*}
DA\left( G^{1}\right) \left( i_{l}\right) =s_{l},\text{ for }l=1,2,3,4.
\end{equation*}%
Then, $E^{1}=\left\{ i_{1}\right\} $ and $U^{1}=\left\{ s_{1}\right\} ,$
that is, $i_{1}$ is eliminated and their assignment is fixed to be $s_{1}$.
Then, since $\left( i_{1},i_{3}\right) \in \succ _{s_{2}}$, $%
Z_{i_{3}}^{2}=\left\{ s_{2}\right\} $. Hence $\emptyset P_{i_{3}}^{2}s_{2}$,
to prevent $i_{3}$ being matched to $s_{2}$.

Second, 
\begin{equation*}
DA\left( G^{2}\right) \left( i_{2}\right) =s_{3},\text{ }DA\left(
G^{2}\right) \left( i_{3}\right) =s_{4},\text{ }DA\left( G^{2}\right) \left(
i_{4}\right) =s_{2}
\end{equation*}%
and the mechanism terminates. Therefore, $EA\left( \succ ,\succ ^{\ast
}\right) =\mu ^{\prime }$.

Thus, the SIC mechanism stops at $\mu $, even though $\mu $ is Pareto
dominated by the other stable matching $\mu ^{\prime }$. By contrast, the
EADA mechanism attains $\mu ^{\prime }$, which is the unique constrained
efficient matching for $\succ $.\newline

Together, Examples 1 and 2 show that the two mechanisms fail in opposite
ways under acyclic but non-transitive priorities. The EADA mechanism may
\textquotedblleft over-improve\textquotedblright\ and produce a matching
that is not stable, whereas the SIC mechanism may \textquotedblleft
under-improve\textquotedblright\ and stop at a stable matching that is
Pareto dominated by another stable matching.

Formally, we have the following two results.

\begin{theorem}
If $\succ \in \mathcal{A}^{\left\vert S\right\vert }$ and $\succ ^{\ast }\in 
\mathcal{E}\left( \succ \right) $, then $EA\left( \succ ,\succ ^{\ast
}\right) $ is not Pareto dominated by any stable matching for $\succ $.
\end{theorem}

The proof of this result is provided in the Appendix.

This fact shows that the EADA mechanism does not suffer from
under-improvement, whereas it does suffer from over-improvement, because, as
explained in Example 1, $EA\left( \succ ,\succ ^{\ast }\right) $ may not be
stable.

\begin{theorem}
If $\succ \in \mathcal{A}^{\left\vert S\right\vert }$ and $\succ ^{\ast }\in 
\mathcal{E}\left( \succ \right) $, then any $\mu \in \mathcal{SI}\left(
\succ ,\succ ^{\ast }\right) $ is stable for $\succ $.\footnote{%
Acyclicity is used here only to guarantee the existence of a linear-order
extension and, hence, of the initial stable matching. Lemma 6 in the
Appendix shows that applying a stable improvement cycle to a stable matching
preserves stability under any asymmetric priority relation.}
\end{theorem}

The proof of this result is provided in the Appendix.

This fact shows that the SIC mechanism does not suffer from
over-improvement, whereas it does suffer from under-improvement, because, as
explained in Example 2, some $\mu \in \mathcal{SI}\left( \succ ,\succ ^{\ast
}\right) $ may be Pareto dominated by another stable matching for $\succ $.

We now return to the consent model introduced in Section 2.2 and consider
the following mechanisms.

First, for $\left( \succ ,C\right) $, we let $\succ ^{C}$ in the definition
of Kitahara and Okumura (2021) as explained above. Second, we derive $\succ
^{\ast }\in \mathcal{E}\left( \succ ^{C}\right) $. Finally, we use either
the EADA mechanism or the SIC mechanism.

We consider the case where $\succ ^{C}\in \mathcal{P}^{\left\vert
S\right\vert }$. Then, since both mechanisms attain constrained efficient
matchings for $\succ ^{C}$, they are partially stable matchings for $\left(
\succ ,C\right) $ that are not Pareto dominated by any partially stable
matching for $\left( \succ ,C\right) $.

However, as is shown above, $\succ _{s}^{C}$ is acyclic but may not be
transitive. Thus, $EA\left( \succ ^{C},\succ ^{\ast }\right) $ may fail to
be a stable matching for $\succ ^{C}$ or equivalently, it may fail to be a
partially stable matching for $\left( \succ ,C\right) $. On the other hand,
by Theorem 3, the result of the SIC mechanism is stable for $\succ ^{C}$ or
equivalently, it is partially stable for $\left( \succ ,C\right) $. However,
some $\mu \in \mathcal{SI}\left( \succ ^{C},\succ ^{\ast }\right) $ may fail
to be a constrained efficient matching for $\succ ^{C};$ that is, $\mu $ may
be Pareto dominated by some other partially stable matching for $\left(
\succ ,C\right) $.

\section*{Concluding Remarks}

This paper has compared the EADA and SIC mechanisms beyond the domain of
weak priority orders. When school priorities are partial orders, both
mechanisms yield constrained efficient matchings. Moreover, by varying the
linear-order extensions, the two mechanisms have the same range of possible
outcomes. Thus, their outcome equivalence continues to hold even when
priorities allow incomparability, provided that transitivity is maintained.

However, this equivalence need not extend beyond transitive priorities.
Under acyclic but non-transitive priorities, the two mechanisms may fail in
opposite ways. The EADA mechanism does not suffer from under-improvement
because its outcome is not Pareto dominated by any stable matching, but may
suffer from over-improvement: its outcome may fail to be stable. By
contrast, the SIC mechanism does not suffer from over-improvement because
its outcome remains stable, but it may suffer from under-improvement: its
outcome may be Pareto dominated by another stable matching. These findings
show that transitivity is sufficient for the equivalence and
constrained-efficiency results established under partial-order priorities,
but it is not necessary. A necessary and sufficient condition for these
results remains unknown.

Acyclic but non-transitive priorities are not merely a theoretical
possibility. They arise in economically relevant settings, including consent
models in which students selectively waive their priority claims in favor of
particular students. Without appropriate restrictions on consent, the
induced priority relations remain acyclic but need not be transitive. Our
results therefore reveal important limitations of both mechanisms in such
environments as explained above. Developing a mechanism that selects a
constrained efficient matching for every profile of acyclic priorities
remains an open question.

\section*{References}

\begin{description}
\item Abdulkadiro\u{g}lu, A., Pathak, P.A., Roth, A.E., 2009.
Strategy-proofness versus efficiency in matching with indifferences:
redesigning the NYC high school match.\ American Economic Review 99,
1954--1978.

\item Abdulkadiro\u{g}lu, A., Pathak, P.A., Roth, A.E., S\"{o}nmez, T. 2005.
The Boston Public School Match.\ American Economic Review 95(2), 368-371.

\item Abdulkadiro\u{g}lu, A., Pathak, P.A., Roth, A.E., S\"{o}nmez, T. 2006.
Changing the Boston Public School Mechanism: Strategy-proofness as equal
access,\ NBER Working paper, 11965.

\item Abdulkadiro\u{g}lu A., S\"{o}nmez, T. 2003. School choice: A mechanism
design approach. American Economic Review 93(3), 729--747.

\item Balinski, M., S\"{o}nmez, T. 1999. A tale of two mechanisms: student
placement,\ Journal of Economic Theory 84(1), 73--94.

\item Bando, K., Imamura, K., Kawase, Y. 2025. Properties of
path-independent choice correspondences and their applications to efficient
and stable matchings. Mimeo Available at arXiv:2502.09265.

\item Cerrone, C., Hermstr\"{u}wer, Y., Kesten, O. 2024. School Choice with
Consent: An Experiment, Economic Journal 134(661), 1760--1805.

\item Che, Y-K., Kim, J., Kojima, F. 2019a. Stable matching in large
economies. Econometrica 87(1), 65-110.

\item Che, Y-K., Kim, J., Kojima, F. 2019b. Weak monotone comparative
statics, Mimeo Available at arXiv:1911.06442

\item Duggan, J. 1999. A General Extension Theorem for Binary Relations,
Journal of Economic Theory 86, 1-16.

\item Dur, U., Gitmez, A., Y\i lmaz, \"{O}. 2019. School choice under
partial fairness,\ Theoretical Economics 14(4), 1309-1346.

\item Erdil, A. 2014. Strategy-proof stochastic assignment. Journal of
Economic Theory 151, 146--162.

\item Erdil, A., Ergin, H. 2008. What's the matter with tie-breaking?
Improving efficiency in school choice,\ American Economic Review 98(3),
669--689.

\item Erdil, A., Kumano, T. 2019. Efficiency and stability under
substitutable priorities with ties. Journal of Economic Theory 184, 104950.

\item Gale D., Shapley L.S. 1962. College admissions and the stability of
marriage. American Mathematical Monthly 69(1):9--15

\item Kesten, O. 2010. School choice with consent. Quarterly Journal of
Economics 125(3), 1297--1348.

\item Kesten, O., \"{U}nver, M.U. 2015. A theory of school-choice lotteries.
Theoretical Economics 10(2), 543--595.

\item Kitahara, M., Okumura, Y. 2020. Stable Improvement Cycles in a
Controlled School Choice. Mimeo Available at SSRN:
https://ssrn.com/abstract=3582421

\item Kitahara, M., Okumura, Y. 2021. Improving Efficiency in School Choice
under Partial Priorities, International Journal of Game Theory 50, 971--987.

\item Kitahara, M., Okumura, Y. 2023. School choice with multiple
priorities. Mimeo, available at arXiv:2308.04780.

\item Kitahara, M., Okumura, Y. 2024. Extensions of Partial Priorities and
Stability in School Choice, Mathematical Social Sciences 131, 1-4.

\item Kuvalekar, A. 2023. Matching with incomplete preferences, Mimeo,
available at https://arxiv.org/abs/2212.02613

\item Tang, Q., Yu, J. 2014. A new perspective on Kesten's school choice
with consent idea. Journal of Economic Theory 154, 543--561.
\end{description}

\section*{Appendix}

This appendix is devoted to the proofs of our main results, Theorems 1--3.
We begin by establishing several technical results that will be used in the
proofs.

We first establish several auxiliary results concerning the EADA procedure.
For each student $i\in I$, let $\kappa \left( i\right) $ denote the round in
which $i$ is eliminated; that is, $i\in E^{\kappa \left( i\right) }$.
Similarly, for each school $s\in S$, let $\kappa \left( s\right) $ denote
the round in which $s$ is eliminated; that is, $s\in U^{\kappa \left(
s\right) }$.

We first introduce several results on the EADA mechanism. For notational
convenience, let $\kappa :I\cup S\rightarrow \mathbb{N}$ satisfying $\kappa
\left( i\right) =k$ for all $i\in E^{k}$ and $\kappa \left( s\right) =k$ for
all $s\in U^{k}$; that is, $\kappa \left( i\right) $ and $\kappa \left(
s\right) $ represent the Rounds in which $i$ and $s$ are eliminated,
respectively.

\begin{lemma}
For all $i\in I$ and all $k=2,3,\cdots ,K,$ $\hat{\mu}^{k}\left( i\right)
R_{i}\hat{\mu}^{k-1}\left( i\right) $.
\end{lemma}

\textbf{Proof.} Fix $k\in \left\{ 2,\cdots ,K\right\} $. Let%
\begin{equation*}
I^{\prime }=\left\{ i\in I\text{ }\left\vert \text{ }\hat{\mu}^{k-1}\left(
i\right) P_{i}\hat{\mu}^{k}\left( i\right) \right. \right\} \text{.}
\end{equation*}%
Suppose not; that is, $I^{\prime }\neq \emptyset $.

Because the assignments of students in $\bigcup\nolimits_{\kappa
=1}^{k-1}E^{\kappa }$ are already fixed before Round $k$, we have $\hat{\mu}%
^{k}\left( i\right) =\hat{\mu}^{k-1}\left( i\right) $ for all $i\in
\bigcup\nolimits_{\kappa =1}^{k-1}E^{\kappa }$. Therefore, we have $%
I^{\prime }\subseteq I^{k}$.

We first show that, for every student $i\in I^{k}$, the school $\hat{\mu}%
^{k-1}\left( i\right) $ does not belong to $Z_{i}^{k}$. Suppose not; that
is, $\hat{\mu}^{k-1}\left( i\right) =s\in Z_{i}^{k}$. Then, by the
definition of $Z_{i}^{k}$, there exists $j\in E^{k-1}$ such that $\left(
j,i\right) \in \succ _{s}$ and $sP_{j}\hat{\mu}^{k-1}\left( j\right) $.
Since $\succ _{s}^{\ast }$ is a linear order extension of $\succ _{s}$, $%
\left( j,i\right) \in \succ _{s}^{\ast }$. Since $i$ is matched to $s$ in $%
DA\left( G^{k-1}\right) $, this means that $j$ strictly prefers $s$ to its
assignment in $DA\left( G^{k-1}\right) ,$ while $\left( j,i\right) \in \succ
_{s}^{\ast }$. Thus, matching $\hat{\mu}^{k-1}$ violates the priority $\succ
_{s}^{\ast }$\ of $j\notin \hat{\mu}^{k-1}\left( s\right) $ over $i\in \hat{%
\mu}^{k-1}\left( s\right) $, contradicting the stability of the DA outcome
in Round $k-1$. Hence $\hat{\mu}^{k-1}\left( i\right) \notin Z_{i}^{k}$ for
every $i\in I^{k}$.

Choose student $i\in I^{\prime }$ who is rejected by $\hat{\mu}^{k-1}\left(
i\right) =DA\left( G^{k-1}\right) \left( i\right) $ in the \textit{earliest
step} of the SPDA in Round $k$ among $I^{\prime }$. Let $s=DA\left(
G^{k-1}\right) \left( i\right) $ and $t$ be that step.

Then, at Step $t$ (of the SPDA in Round $k$), school $s$ temporarily accepts 
$q_{s}$ students, each of whom has higher priority than $i$ at $s$ under $%
\succ _{s}^{\ast }$ at step $t$. Since $i\in \hat{\mu}^{k-1}\left( s\right) $%
, at least one of those temporarily accepted $q_{s}$ students, say $j$, who
is not matched to $s$ in $\hat{\mu}^{k-1}\left( s\right) $. Thus, $\left(
j,i\right) \in \succ _{s}^{\ast }$.

Since $j$ is temporarily accepted by $s$ in Step $t$ of the SPDA in Round $k$%
, $sP_{j}^{k}\emptyset $ and hence $s\notin Z_{j}^{k}$. Moreover, since $i$
is matched to $s$ in $\hat{\mu}^{k-1}$ and $\left( j,i\right) \in \succ
_{s}^{\ast },$ $\hat{\mu}^{k-1}(j)P_{j}^{k-1}s$. Since $\hat{\mu}%
^{k-1}(j)\notin Z_{j}^{k}$ as shown above, both $s$ and $\hat{\mu}^{k-1}(j)$
are not in $Z_{j}^{k}$. Their relative order is unchanged when preferences
are modified from $P_{j}^{k-1}$ to $P_{j}^{k}$. Hence $\hat{\mu}%
^{k-1}(j)P_{j}^{k}s$. Moreover, $sR_{j}^{k}\hat{\mu}^{k}(j)R_{j}^{k}%
\emptyset $. Hence $\hat{\mu}^{k-1}(j)P_{j}^{k}\hat{\mu}^{k}(j)R_{j}^{k}%
\emptyset $ and $\hat{\mu}^{k-1}(j)P_{j}\hat{\mu}^{k}(j)R_{j}\emptyset $.

Therefore $j\in I^{\prime }$. Since $j$ is temporarily accepted by $s$ in
Step $t$, $j$ is rejected by $DA\left( G^{k-1}\right) \left( j\right) $
before Step $t$. Since $i$ is the student who is firstly rejected by $%
DA\left( G^{k-1}\right) \left( i\right) $, this is a contradiction. \textbf{%
Q.E.D. }\newline

\begin{lemma}
(Erdil 2014, Reshuffling Lemma) If $v$ weakly Pareto dominates an
individually rational and nonwasteful matching denoted by $\mu $, then $v$
is also individually rational and nonwasteful and moreover, 
\begin{equation}
\left\vert v\left( s\right) \right\vert =\left\vert \mu \left( s\right)
\right\vert \text{ for all }s\in S\text{.}  \label{x}
\end{equation}
\end{lemma}

\textbf{Proof. }Erdil (2014) shows (\ref{x}) and the individual rationality
of $v$ is trivial.

We show $v$ is nonwasteful. Since $\mu $ is nonwasteful, $\left\vert \mu
\left( s\right) \right\vert <q_{s}$ implies $\mu \left( i\right) R_{i}s$ for
all $i\in I$. By (\ref{x}), $\left\vert v\left( s\right) \right\vert <q_{s}$
implies $\mu \left( i\right) R_{i}s$. Moreover, since $v$ Pareto dominates $%
\mu $, $v\left( i\right) R_{i}s$. \textbf{Q.E.D.}

\begin{lemma}
For all $k=1,\cdots ,K$, $\hat{\mu}^{k}$ is individually rational and
nonwasteful. Moreover, $\left\vert \hat{\mu}^{1}\left( s\right) \right\vert
=\cdots =\left\vert \hat{\mu}^{K}\left( s\right) \right\vert $ for all $s\in
S$.
\end{lemma}

\textbf{Proof.} First, since $\hat{\mu}^{1}$ is stable for $\succ $ (by
Remark 2), it is individually rational and nonwasteful. Next, by Lemma 1, $%
\hat{\mu}^{k}$ weakly Pareto dominates $\hat{\mu}^{k-1}$.

Therefore, applying the Reshuffling Lemma inductively, we obtain that $\hat{%
\mu}^{k}$ is individually rational and nonwasteful for every $k=2,\cdots ,K$.

Moreover, the first result and the Reshuffling Lemma directly imply $%
\left\vert \hat{\mu}^{1}\left( s\right) \right\vert =\cdots =\left\vert \hat{%
\mu}^{K}\left( s\right) \right\vert $ for all $s\in S$. \textbf{Q.E.D.}

\begin{lemma}
If $\succ \in \mathcal{P}^{\left\vert S\right\vert }$ and $\hat{\mu}^{k-1}$
is stable for $\succ $, then$\ \hat{\mu}^{k}$ is also stable for $\succ $.
\end{lemma}

\textbf{Proof.} By Lemma 3, $\hat{\mu}^{k}$ is individually rational and
nonwasteful. We show the fairness of $\hat{\mu}^{k}$.

Suppose not; that is, $\hat{\mu}^{k}$ is not fair. Then, there are $i$ and $%
j $ such that 
\begin{equation*}
\hat{\mu}^{k}\left( j\right) =sP_{i}\hat{\mu}^{k}\left( i\right) \ \text{and 
}\left( i,j\right) \in \succ _{s}.
\end{equation*}%
Since $\hat{\mu}^{1}=\mu ^{1}$ is stable for $\succ $ by Remark 2, $\hat{\mu}%
^{1},\hat{\mu}^{2},\cdots ,\hat{\mu}^{k-1}$ are stable for $\succ $.

First, we show $i\in I^{k}$; that is, $i$ has not been eliminated before
Round $k$. Suppose not; that is, $i\in \bigcup\nolimits_{\kappa
=1}^{k-1}E^{\kappa }$. If $j\in \bigcup\nolimits_{\kappa =1}^{k-1}E^{\kappa
}$ as well, then $\hat{\mu}^{k}\left( j\right) =\hat{\mu}^{k-1}\left(
j\right) =s$ and $\hat{\mu}^{k}\left( i\right) =\hat{\mu}^{k-1}\left(
i\right) $. However, these contradict the stability of $\hat{\mu}^{k-1}$.
Hence $j\in I^{k}\,$and $s\in S^{k}$. Since $i$ was eliminated in an earlier
round, $\left( i,j\right) \in \succ _{s}$ and $sP_{i}\hat{\mu}^{k}\left(
i\right) =\hat{\mu}^{k-1}\left( i\right) $. By the definition of the EADA
mechanism, $\emptyset P_{j}^{k}s$. Hence $j$ cannot be matched to $s$ in
Round $k$, contradicting $\hat{\mu}^{k}\left( j\right) =s$. Thus, $i\in
I^{k} $.

Next, we consider two cases.

First, suppose $sP_{i}^{k}\emptyset $. Since $i\in I^{k},$ $s$ is acceptable
to $i$ in Round $k$. We show $s\in S^{k}$. Suppose not. Then, $s$ was
eliminated at an earlier Round $\kappa <k$ and thus it is underdemanded in
the Round $\kappa $. Since $\hat{\mu}^{\kappa }$ is nonwasteful by Lemma 3,
we have $\hat{\mu}^{\kappa }\left( i\right) R_{i}s$. On the other hand, by
Lemma 1, 
\begin{equation*}
sP_{i}\hat{\mu}^{k}\left( i\right) R_{i}\hat{\mu}^{k-1}\left( i\right)
R_{i}\cdots R_{i}\hat{\mu}^{\kappa }\left( i\right) .
\end{equation*}
Hence $\hat{\mu}^{k}\left( i\right) R_{i}s$ contradicting $sP_{i}\hat{\mu}%
^{k}\left( i\right) $. Therefore, $s\in S^{k}$.

Since $\hat{\mu}^{k}\left( j\right) =s$ and $s\in S^{k}$, $j\in I^{k}$.
Moreover, since $sP_{i}^{k}\emptyset $, we have $s\notin Z_{i}^{k}$. Since $%
\hat{\mu}^{k}\left( i\right) \notin Z_{i}^{k}$ as well, we have 
\begin{equation*}
sP_{i}^{k}\hat{\mu}^{k}\left( i\right) \text{.}
\end{equation*}%
Since $\left( i,j\right) \in \succ _{s}$, $\left( i,j\right) \in \succ
_{s}^{\ast }$. Since both $i$ and $j$ apply to $s$ in the SPDA of Round $k$, 
$\left( j,i\right) \in \succ _{s}^{\ast }$. However, since $\succ _{s}^{\ast
}$ is an extension of $\succ _{s}$, $\left( j,i\right) \in \succ _{s}^{\ast
} $ contradicts $\left( i,j\right) \in \succ _{s}$.

Second, suppose $\emptyset P_{i}^{k}s$. Since $i\in I^{k}$ and $sP_{i}\hat{%
\mu}^{k}\left( i\right) $, this means $s\in Z_{i}^{k}$. Then, there is $%
i^{\prime }\in \bigcup\nolimits_{\kappa =1}^{k-1}E^{\kappa }$ such that 
\begin{equation*}
sP_{i^{\prime }}\hat{\mu}^{k}\left( i^{\prime }\right) =\hat{\mu}%
^{k-1}\left( i^{\prime }\right) \text{ and }\left( i^{\prime },i\right) \in
\succ _{s}.
\end{equation*}%
By the transitivity of $\succ _{s}$, $\left( i^{\prime },i\right) \in \succ
_{s}$ and $\left( i,j\right) \in \succ _{s}$ imply $\left( i^{\prime
},j\right) \in \succ _{s}$.\footnote{%
When $\succ _{s}$ is not transitive for some $s\in S$, \ $\hat{\mu}^{k}$ may
not be fair for $\succ $ even if $\hat{\mu}^{k-1}$ is fair for $\succ $.}
However, by the first part of this proof, $i^{\prime }\in
\bigcup\nolimits_{\kappa =1}^{k-1}E^{\kappa }$, $\left( i^{\prime
},j\right) \in \succ _{s}$ and $\hat{\mu}^{k}\left( j\right) =sP_{i^{\prime
}}\hat{\mu}^{k}\left( i^{\prime }\right) $ are not compatible. This
contradiction establishes the fairness of $\hat{\mu}^{k}$, completing the
proof. \textbf{Q.E.D. }\newline

\begin{lemma}
If $\succ \in \mathcal{P}^{\left\vert S\right\vert }$, then$\ \hat{\mu}^{k}$
is stable for $\succ $ and all $k=1,\cdots ,K$.
\end{lemma}

\textbf{Proof.} First, since $\hat{\mu}^{1}$ is stable for $\succ $ (by
Remark 2), it is individually rational and nonwasteful. By Lemma 4, $\hat{\mu%
}^{2}$ is also stable for $\succ $. Hence, by induction, $\hat{\mu}^{k}$ is
stable for $\succ $ and all $k=3,\cdots ,K$. \textbf{Q.E.D.}

\subparagraph*{Proof of Theorem 2}

Suppose not; that is, there is a stable matching $\nu $ for $\succ $ that
Pareto dominates $\hat{\mu}^{K}\left( =EA\left( \succ ,\succ ^{\ast }\right)
\right) $. Let 
\begin{equation*}
I^{\ast }=\left\{ i\in I\text{ }\left\vert \text{ }\nu \left( i\right) P_{i}%
\hat{\mu}^{K}\left( i\right) \right. \right\} .
\end{equation*}%
Then, $I^{\ast }\neq \emptyset $.

First, we show that there exist $i\in I^{\ast }$ such that $\kappa \left(
\nu \left( i\right) \right) >\kappa \left( i\right) $; that is, the round in
which the school $\nu \left( i\right) $ is eliminated is later than the
round in which student $i$ is eliminated.

We choose $i^{\ast }\in I^{\ast }$ such that $\kappa \left( i^{\ast }\right)
\leq \kappa \left( i\right) $ for all $i\in I^{\ast };$ that is, $i^{\ast }$
is one of the first students in $I^{\ast }$ to be eliminated. Moreover, let $%
k=\kappa \left( i^{\ast }\right) $ and $s=\nu \left( i^{\ast }\right) $.

We show $\kappa \left( s\right) >k$. Suppose not; that is, $\kappa \left(
s\right) \leq k$. Since $sP_{i^{\ast }}\hat{\mu}^{K}\left( i^{\ast }\right) $%
, and by Lemma 1, $\hat{\mu}^{K}\left( i^{\ast }\right) R_{i^{\ast }}\hat{\mu%
}^{\kappa \left( s\right) }\left( i^{\ast }\right) $ and therefore $%
sP_{i^{\ast }}\hat{\mu}^{\kappa \left( s\right) }\left( i^{\ast }\right) $.
Since $\kappa \left( s\right) \leq k=\kappa \left( i^{\ast }\right) $,
student $i^{\ast }$ is active in Round $\kappa \left( s\right) $. Since $s$
is underdemanded in Round $\kappa \left( s\right) $, $i^{\ast }$ cannot
apply to $s$ in that round. Hence $s$ must have been made unacceptable to $%
i^{\ast }$, so there exists some student $i^{\prime }$ such that $\kappa
\left( i^{\prime }\right) <\kappa \left( s\right) $, $sP_{i^{\prime }}\hat{%
\mu}^{K}\left( i^{\prime }\right) $ and $\left( i^{\prime },i^{\ast }\right)
\in \succ _{s}$. Since $\nu $ is stable for $\succ $, $\nu \left( i^{\prime
}\right) R_{i^{\prime }}sP_{i^{\prime }}\hat{\mu}^{K}\left( i^{\prime
}\right) $ and thus $i^{\prime }\in I^{\ast }$. However, this contradicts
the fact that $i^{\ast }$ is one of the first students in $I^{\ast }$ to be
eliminated.

Therefore, $\kappa \left( s\right) >k$. Since $i^{\ast }$ is eliminated in
Round $k$, we have 
\begin{equation*}
\hat{\mu}^{k}\left( i^{\ast }\right) \in \left( S\cup \left\{ \emptyset
\right\} \right) \setminus S^{k+1},
\end{equation*}%
while $s\in S^{k+1}$. If we define two subsets of students; 
\begin{eqnarray*}
I_{1} &=&\left\{ i\in I\text{ }\left\vert \text{ }\nu \left( i\right) \in
S^{k+1}\text{ and }\hat{\mu}^{k}\left( i\right) \in \left( S\cup \left\{
\emptyset \right\} \right) \setminus S^{k+1}\right. \right\} , \\
I_{2} &=&\left\{ i\in I\text{ }\left\vert \text{ }\nu \left( i\right) \in
\left( S\cup \left\{ \emptyset \right\} \right) \setminus S^{k+1}\text{ and }%
\hat{\mu}^{k}\left( i\right) \in S^{k+1}\right. \right\} ,
\end{eqnarray*}%
then, $i^{\ast }\in I_{1}$ and hence $I_{1}$ is not empty.

Next, we show that $I_{2}$ is empty. Suppose not; that is, $i^{\prime }\in
I_{2}$. Since $\hat{\mu}^{k}\left( i^{\prime }\right) \in S^{k+1}$, $%
i^{\prime }$ is still active after Round $k$. Since $\nu $ Pareto dominates $%
\hat{\mu}^{K}$, and by Lemma 1, $\nu $ also Pareto dominates $\hat{\mu}^{k}$%
. Therefore $\nu \left( i^{\prime }\right) \neq \emptyset $ and, let $%
s^{\prime }=\nu \left( i^{\prime }\right) $ and $k^{\prime }=\kappa \left(
\nu \left( i^{\prime }\right) \right) \left( \leq k\right) $.

Then, $s^{\prime }$ is underdemanded in Round $k^{\prime }.$ Moreover, 
\begin{equation*}
s^{\prime }P_{i^{\prime }}\hat{\mu}^{k}\left( i^{\prime }\right)
R_{i^{\prime }}\hat{\mu}^{k^{\prime }}\left( i^{\prime }\right) \text{.}
\end{equation*}%
Since $i^{\prime }$ is active and $s^{\prime }$ is underdemanded in Round $%
k^{\prime }$, $i^{\prime }$ does not apply to $s^{\prime }$.

Therefore, $s^{\prime }P_{i^{\prime }}\emptyset $ but$\ \emptyset
P_{i^{\prime }}^{k^{\prime }}s^{\prime }$. Thus, there is $i^{\prime \prime
} $ such that $\kappa \left( i^{\prime \prime }\right) <k^{\prime },$ $%
s^{\prime }P_{i^{\prime \prime }}\hat{\mu}^{k}\left( i^{\prime \prime
}\right) $ and $\left( i^{\prime \prime },i^{\prime }\right) \in \succ
_{s^{\prime }}$. Since $\nu $ is fair for $\succ $, $\nu \left( i^{\prime
\prime }\right) R_{i^{\prime \prime }}s^{\prime }$ and thus $i^{\prime
\prime }\in I^{\ast }$. However, $\kappa \left( i^{\prime \prime }\right)
<k^{\prime }\leq k=\kappa \left( i^{\ast }\right) $, contradicting the
choice of $i^{\ast }$. Therefore, $I_{2}=\emptyset $.

Then, since $\left\vert I_{1}\right\vert \geq 1$, there is at least one
student $i$ such that $\nu \left( i\right) \in S^{k+1}$ and $\hat{\mu}%
^{k}\left( i\right) \notin S^{k+1}$. However, since $\left\vert
I_{2}\right\vert =0$, there is no student $i$ such that $\nu \left( i\right)
\notin S^{k+1}$ and $\hat{\mu}^{k}\left( i\right) \in S^{k+1}$. Thus, there
is $s\in S^{k+1}$ such that $\left\vert \hat{\mu}^{k}\left( s\right)
\right\vert <\left\vert \nu \left( s\right) \right\vert \leq q_{s}$. By
Lemma 3, $\hat{\mu}^{k}$ is individually rational and nonwasteful, and
moreover, $\hat{\mu}^{k}$ is Pareto dominated by $\nu $. By the Reshuffling
Lemma, $\left\vert \hat{\mu}^{k}\left( s\right) \right\vert =\left\vert \nu
\left( s\right) \right\vert $ which is a contradiction. Therefore, no stable
matching for $\succ $ Pareto dominates $\hat{\mu}^{K}=EA\left( \succ ,\succ
^{\ast }\right) $. \textbf{Q.E.D.}

\subparagraph*{Proof of Theorem 1}

By Lemma 5, $EA\left( \succ ,\succ ^{\ast }\right) =\hat{\mu}^{K}$ is stable
for $\succ $. By Theorem 2, $EA\left( \succ ,\succ ^{\ast }\right) =\hat{\mu}%
^{K}$ is not Pareto dominated by any stable matching $\mu $ for $\succ $.
Therefore, $EA\left( \succ ,\succ ^{\ast }\right) $ is constrained efficient
for $\succ $. \textbf{Q.E.D.}

\subparagraph*{Proof of Theorem 3}

First, we prove the following result.

\begin{lemma}
Let $\mu $ be a stable matching for $\succ $ and $\phi $ be a stable
improvement cycle of $\mu $. Then, $\phi \circ \mu $ is also stable for $%
\succ $.
\end{lemma}

\textbf{Proof.} Let $\phi =\left\{ i_{1}i_{2},\cdots ,i_{q}i_{q+1}\right\} $
and $\nu =\phi \circ \mu $. Since $\mu $ is individually rational and
nonwasteful, the Reshuffling Lemma implies that $\nu $ is also individually
rational and nonwasteful.

We show that $\nu $ is fair for $\succ $. Suppose not; that is, there are $%
i,j,s$ such that $sP_{i}\nu \left( i\right) ,\ j\in \nu \left( s\right) ,$
and $\left( i,j\right) \in \succ _{s}$. Then, $\nu \left( i\right) R_{i}\mu
\left( i\right) $ is satisfied and thus $sP_{i}\mu \left( i\right) $.

First, suppose $j\notin \left\{ i_{1},\cdots ,i_{q}\right\} $. Then, $\nu
\left( j\right) =\mu \left( j\right) =s$ and $sP_{i}\mu \left( i\right) ,$
and $\left( i,j\right) \in \succ _{s}$. These facts contradict that $\mu $
is fair for $\succ $.

Second, suppose $j\in \left\{ i_{1},\cdots ,i_{q}\right\} $. Without loss of
generality, we assume $j=i_{1}$. Then, $j\in X_{i_{2}}\left( \mu ,\succ
\right) $ and $s$ $=\mu \left( i_{2}\right) $. Since $sP_{i}\mu \left(
i\right) $, $i,j\in D_{i_{2}}\left( \mu \right) $. Moreover, since $\left(
i,j\right) \in \succ _{s}$, $j\in Y_{i_{2}}\left( \mu ,\succ \right) ,$ but
this contradicts $j\in X_{i_{2}}\left( \mu ,\succ \right) $. \textbf{Q.E.D.}

First, since $\succ \in \mathcal{A}^{\left\vert S\right\vert }$, by Remarks
1 and 2, $\mu ^{1}$ is stable for $\succ $. By Lemma 6, we can inductively
show that $\mu ^{2},\mu ^{3}\cdots $ are also stable for $\succ $. \textbf{%
Q.E.D.}

\end{document}